\documentclass{article}
\usepackage{ijcai25}

\usepackage{times}
\usepackage{soul}
\usepackage{url}
\usepackage[hidelinks]{hyperref}
\usepackage[utf8]{inputenc}
\usepackage[small]{caption}
\usepackage{graphicx}
\usepackage{amsmath}
\usepackage{amsthm}
\usepackage{amssymb}
\usepackage{amsfonts}
\usepackage{booktabs}
\usepackage{algorithm}
\usepackage[noend]{algorithmic}
\usepackage[switch]{lineno}

\usepackage{natbib}
\usepackage{cleveref}
\usepackage{bm}

\newtheorem{theorem}{Theorem}
\newtheorem{lemma}{Lemma}
\newtheorem{definition}{Definition}
\newtheorem{proposition}{Proposition}

\newcommand{\eps}{\varepsilon}
\newcommand{\cD}{\mathcal{D}}
\newcommand{\E}{\mathbb{E}}

\usepackage[color={red!100!green!33},textsize=scriptsize]{todonotes}

\title{Asymptotic Analysis of Weighted Fair Division}

\author{
Pasin Manurangsi$^1$
\and
Warut Suksompong$^2$\and
Tomohiko Yokoyama$^3$\\
\affiliations
$^1$Google Research, Thailand\\
$^2$National University of Singapore, Singapore\\
$^3$The University of Tokyo, Japan
}

\allowdisplaybreaks

\begin{document}

\maketitle

\begin{abstract}
Several resource allocation settings involve agents with unequal entitlements represented by weights.
We analyze weighted fair division from an asymptotic perspective: if $m$ items are divided among $n$ agents whose utilities are independently sampled from a probability distribution, when is it likely that a fair allocation exist?
We show that if the ratio between the weights is bounded, a weighted envy-free allocation exists with high probability provided that $m = \Omega(n\log n/\log\log n)$, generalizing a prior unweighted result.
For weighted proportionality, we establish a sharp threshold of $m = n/(1-\mu)$ for the transition from non-existence to existence, where $\mu\in (0,1)$ denotes the mean of the distribution.
In addition, we prove that for two agents, a weighted envy-free (and weighted proportional) allocation is likely to exist if $m = \omega(\sqrt{r})$, where $r$ denotes the ratio between the two weights.
\end{abstract}

\section{Introduction}

The fair allocation of scarce resources is a fundamental problem in economics and has received substantial attention recently in computer science \citep{BramsTa96,RobertsonWe98,Moulin19}.
Research in fair division has sought to quantify fairness via precise mathematical definitions, among which two of the most important are envy-freeness and proportionality.
An allocation is said to be \emph{envy-free} if each agent values her own bundle at least as much as any other agent's bundle \citep{Foley67,Varian74}.
It is called \emph{proportional} if every agent values her bundle at least $1/n$ of her value for the entire set of resources, where $n$ denotes the number of agents \citep{Steinhaus48}.

When allocating indivisible items, like houses, cars, or musical instruments, an envy-free or proportional allocation may not exist---this is the case, for example, when all items are valuable and the number of items $m$ is smaller than $n$.
In light of this, a line of work has focused on the question of \emph{when} a fair allocation is likely to exist if the agents' utilities for the items are drawn independently from a probability distribution.\footnote{We survey this and another related line of work in \Cref{sec:related-work}.}
For instance, \citet{ManurangsiSu20,ManurangsiSu21} showed that an envy-free allocation exists with high probability\footnote{That is, the probability that it exists converges to $1$ as $n\rightarrow\infty$.} provided that $m = \Omega(n\log n / \log\log n)$.
Not only is this bound tight, but it can also be achieved by a simple round-robin algorithm which lets the agents take turns picking their favorite item from the remaining items until the items run out.
The same authors also proved that a proportional allocation is likely to exist as long as $m\ge n$.

The aforementioned results, like most of the work in fair division, assume that all agents have the same entitlement to the resource.
In recent years, a number of researchers have explored a more general framework where agents may have varying entitlements represented by weights \citep{FarhadiGhHa19,AzizMoSa20,AzizGaMi23,ChakrabortyIgSu21,ChakrabortySeSu24,HoeferScVa24,SpringerHaYa24}.
This broader framework allows us to model scenarios such as distributing resources among communities, where larger communities naturally deserve a larger portion of the resource, as well as dividing inheritance, where closer relatives typically receive a greater share of the bequest.
Fortunately, both proportionality and envy-freeness can be extended to the weighted setting in an intuitive manner.
As an example, if there are three agents with weights $1$, $3$, and $6$, then weighted proportionality requires the first agent to receive at least $1/10$ of her value for the entire set of items, while weighted envy-freeness stipulates that the second agent should not value the third agent's bundle more than twice the value of her own bundle.

In addition to its inherent motivation, interest in weighted fair division stems from the fact that it exhibits several differences and brings a range of new challenges compared to its unweighted counterpart.
For instance, in the absence of weights, it is trivial to show that the round-robin algorithm guarantees \emph{envy-freeness up to one item (EF1)}, meaning that if an agent envies another agent, then the envy disappears upon removing some item in the latter agent's bundle.
By contrast, while it is possible to extend the round-robin algorithm to incorporate weights and prove that the resulting algorithm ensures \emph{weighted EF1 (WEF1)}, doing so is far less straightforward \citep{ChakrabortyIgSu21,WuZhZh23}.
In this paper, we shall analyze weighted fair division from an asymptotic perspective.
Do the same relations between $m$ and $n$ continue to guarantee the existence of fair allocations when agents may have different weights?

\subsection{Our Results}

We assume that each agent's utility for each item is drawn independently from a non-atomic distribution $\cD$ over $[0,1]$.
For most results, we also assume that $\cD$ is \emph{PDF-bounded}, that is, its probability distribution function is bounded between $\alpha$ and $\beta$ throughout $[0,1]$ for some constants $\alpha,\beta > 0$.
All of our existence results come with polynomial-time algorithms.

In \Cref{sec:WEF,sec:WPROP}, we consider the setting where the ratio between the weights is bounded.
In \Cref{sec:WEF}, we show that if $m = \Omega(n\log n/\log\log n)$, then the weighted picking sequence algorithm of \citet{ChakrabortyIgSu21} produces a weighted envy-free allocation with high probability; this generalizes the corresponding unweighted result by \citet{ManurangsiSu21}.
Along the way, we provide an alternative proof that the allocation returned by this algorithm always satisfies WEF1; this proof is arguably simpler than existing proofs \citep{ChakrabortyIgSu21,WuZhZh23} and may therefore be of independent interest.
In \Cref{sec:WPROP}, we turn our attention to weighted proportionality and prove that, interestingly, the transition from non-existence to existence depends on the mean of the distribution~$\cD$: it occurs at $m = n/(1-\mu)$, where $\mu\in (0,1)$ denotes the mean of $\cD$.
This implies that achieving (weighted) proportionality is more difficult than in the unweighted setting, for which the threshold is $m = n$ \citep{ManurangsiSu21}.

In \Cref{sec:two-agents}, we focus on the case of two agents but allow the ratio $r\ge 1$ between the agents' weights to grow; in this case, weighted envy-freeness and weighted proportionality are equivalent.
We show that if $m = \omega(\sqrt{r})$, then the probability that a weighted envy-free (and therefore weighted proportional) allocation exists approaches~$1$ as $r\rightarrow\infty$.
We also establish the tightness of this bound.

\subsection{Further Related Work}
\label{sec:related-work}

The asymptotic analysis of fair division was initiated by \citet{DickersonGoKa14}, who showed that an envy-free allocation is likely to exist if $m = \Omega(n\log n)$, but unlikely to exist when $m = n + o(n)$.
\citet{ManurangsiSu20} strengthened these results by exhibiting that existence is likely as long as $m\ge 2n$ if $m$ is divisible by $n$, but unlikely even when $m = \Theta(n\log n/\log\log n)$ if $m$ is not ``almost divisible'' by $n$.
The gap in the non-divisible case was closed by \citet{ManurangsiSu21}, who demonstrated the existence with high probability when $m = \Omega(n\log n/\log\log n)$ via the round-robin algorithm.
The same authors also proved that a proportional allocation is likely to exist when $m\ge n$, generalizing earlier results by \citet{Suksompong16} and \citet{AmanatidisMaNi17}.
\citet{ManurangsiSu17} studied envy-freeness when the agents are partitioned into groups, while \citet{YokoyamaIg25} performed an asymptotic analysis of ``class envy-free'' matchings.
Beyond envy-freeness and proportionality, \citet{KurokawaPrWa16} and \citet{FarhadiGhHa19} considered \emph{maximin share fairness}---a weaker notion than proportionality---the latter authors also allowing unequal entitlements.
\citet{BaiGo22} considered an extension where different agents may have different distributions, whereas \citet{BaiFeGo22} investigated a ``smoothed utility model'' in which each agent has a base utility for each item and this utility is ``boosted'' with some probability.
\citet{BenadeHaPs24} examined a stochastic model where the agents' utility functions can be non-additive.
\citet{ManurangsiSu25} conducted an asymptotic analysis of fair division for chores (i.e., items that yield negative utilities).

While weighted fair allocation has previously been studied for \emph{divisible} items \citep{Segalhalevi19,CrewNaSp20,CsehFl20}, following the broader trend in fair division, it has been intensively examined in the context of \emph{indivisible} items over the last few years.
Several authors have proposed and studied variants of envy-freeness, proportionality, and maximin share in the setting with entitlements \citep{FarhadiGhHa19,AzizMoSa20,BabaioffNiTa21,BabaioffEzFe24,ChakrabortyIgSu21,ChakrabortyScSu21,ChakrabortySeSu24,WuZhZh23,SpringerHaYa24,MontanariScSu25}.
For an extensive overview of weighted fair division, we refer to the survey by \citet{Suksompong25}.

\section{Preliminaries}

Let $[k] := \{1,2,\ldots,k\}$ for any positive integer $k$.
We want to allocate a set $M=[m]$ of indivisible items among a set $N=[n]$ of agents.
Each agent $i\in N$ has a utility $u_i(g)\ge 0$ for each item $g\in M$, and a positive weight $w_i\in \mathbb{R}_{>0}$. 
Let $w_{\max}$, $w_{\min}$, and $W$ denote the maximum, minimum, and sum of the agents' weights, respectively. 
We assume that utilities are \emph{additive}, i.e., $u_i(S) = \sum_{g\in S} u_i(g)$ for any subset $S\subseteq M$.
We refer to any subset of items as a \emph{bundle}.
An allocation $A=(A_1,A_2,\ldots,A_n)$ is a partition of $M$ into $n$ bundles, where $A_i$ represents the bundle allocated to agent $i\in N$.

We consider a number of fairness notions.
An allocation~$A$ is said to be \emph{weighted envy-free (WEF)} if for every pair of agents $i,j\in N$, it holds that $\frac{u_i(A_i)}{w_i} \ge \frac{u_i(A_j)}{w_j}$.
As a relaxation of WEF, an allocation~$A$ is \emph{weighted envy-free up to one item (WEF1)} if for all $i,j\in N$ with $A_j\ne\emptyset$, there exists an item $g\in A_j$ such that $\frac{u_i(A_i)}{w_i} \ge \frac{u_i(A_j \setminus \{g\})}{w_j}$.
An allocation $A$ is said to be \emph{weighted proportional (WPROP)} if $u_i(A_i) \ge \frac{w_i}{W}\cdot u_i(M)$ for all $i\in N$.
Note that every allocation that satisfies WEF also satisfies WPROP.

For each agent $i\in N$ and each item $g\in M$, we assume that the utility $u_i(g)$ is drawn independently from a probability distribution $\mathcal{D}$ over $[0,1]$.
A distribution is \emph{non-atomic} if it assigns zero probability to any single point.
Let $f_{\mathcal{D}}$ denote the probability density function (PDF) of $\mathcal{D}$. 
For $\alpha,\beta > 0$,
we say that a distribution $\mathcal{D}$ is $(\alpha,\beta)$-\emph{PDF-bounded} if it is non-atomic and $\alpha \le f_{\mathcal{D}}(x) \le \beta$ for all $x\in [0,1]$ \citep{ManurangsiSu21}. 
A distribution $\mathcal{D}$ is \emph{PDF-bounded} if it is $(\alpha,\beta)$-PDF-bounded for some constants $\alpha,\beta > 0$. 
A random event is said to occur \emph{with high probability} if its probability approaches $1$ as $n\to \infty$.
We use $\log$ to denote the natural logarithm (with base $e$).

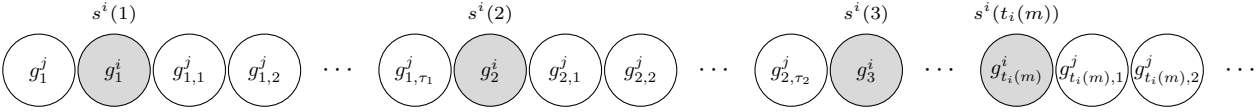
\begin{figure*}[t]
    \begin{tikzpicture}
        \centering
        \newcommand{\drawgroup}[2]{
            \foreach \x/\label/\isI/\num in {#2}
            {   
                \if\isI1
                    \draw[fill=gray!30] (#1+\x,0) circle (0.48);
                \else
                    \draw[fill=white] (#1+\x,0) circle (0.48);
                \fi
                \node[scale=0.8] at (#1+\x,0) {\label};
                \if\isI1
                    \node[above] at (#1+\x,0.5) {\scriptsize $s^i(\num)$};
                \fi
            }
        }
        \drawgroup{0}{0/$g_{1}^j$/0/0, 1/$g_{1}^i$/1/1, 2/$g_{1,1}^j$/0/0, 3/$g_{1,2}^j$/0/0}
        \node at (4,0) {$\cdots$};
        \drawgroup{5}{0/$g_{1,\tau_1}^j$/0/0, 1/$g_{2}^i$/1/2, 2/$g_{2,1}^j$/0/0, 3/$g_{2,2}^j$/0/0}
        \node at (9,0) {$\cdots$};
        \drawgroup{10}{0/$g_{2,\tau_2}^j$/0/0, 1/$g_{3}^i$/1/3}
        \node at (12,0) {$\cdots$};
        \drawgroup{13}{0/$g_{t_i(m)}^i$/1/{t_i(m)}, 1/$g_{t_i(m),1}^j$/0/0, 2/$g_{t_i(m),2}^j$/0/0}
        \node at (16,0) {$\cdots$};
    \end{tikzpicture}
    \caption{An example of the picks by agents $i$ and $j$ in the proof of \Cref{prop:WEF1-alternative}}
    \label{fig:item-order}
\end{figure*}

We now state two lemmas that will be useful for our purposes.
In our analysis of the weighted picking sequence algorithm, we will apply Abel's summation formula, which allows us to rewrite a sum of products in a different form.
\begin{lemma}[Abel's summation formula]\label{lemma:Abel_summation}
    For any sequences of real numbers $(a_1,a_2,\ldots, a_n)$ and $(b_1,b_2,\ldots,b_n)$, we have 
    \[
        \sum_{i=1}^n a_i b_i = a_n \sum_{i=1}^n b_i + \sum_{i=1}^{n-1} \left((a_i-a_{i+1}) \sum_{i'=1}^i b_{i'}\right).
    \]
\end{lemma}
The following Chernoff bound is a standard concentration bound which will be used when we analyze the asymptotic existence of WPROP allocations and the case of two agents.
\begin{lemma}[Chernoff bound]\label{lemma:Chernoff}
    Let $X_1,X_2,\ldots,X_d$ be independent random variables such that $X_i \in [0,1]$ for all $i\in [d]$, and let $X = \sum_{i=1}^d X_i$.
    Then, for any $\delta > 0$, we have
    \begin{enumerate}
        \item $\mathrm{Pr}\left[X \ge (1+\delta) \mathbb{E}[X] \right] \le \mathrm{exp}\left({-\frac{\delta^2 }{2 + \delta}} \mathbb{E}[X]\right)$, and
        \item $\mathrm{Pr}\left[X \le (1-\delta) \mathbb{E}[X] \right] \le \mathrm{exp}\left({-\frac{ \delta^2 }{2}}\mathbb{E}[X]\right)$.
    \end{enumerate}
\end{lemma}

\section{Weighted Envy-Freeness}
\label{sec:WEF}

In this section, we consider weighted envy-freeness. 
We provide an alternative proof that the weighted picking sequence algorithm of \citet{ChakrabortyIgSu21} always outputs a WEF1 allocation.
Our proof also lends itself to the asymptotic analysis of WEF, which we present as Theorem~\ref{thm:WeightedPicking}.

\subsection{Alternative Proof of WEF1}
\label{sec:WEF1-alternative}

We first describe the weighted picking sequence algorithm (Algorithm~\ref{alg:WPS-mechanism}). 
Define a \emph{step} as an iteration of the while-loop (lines~\ref{code:for-start}--\ref{code:wps}).
Here, $t_i$ represents the number of items that agent $i$ has picked so far, and in each step, an agent $i$ who minimizes $t_i/w_i$ picks her favorite item from the remaining items.
\begin{algorithm}[H]
    \caption{Weighted Picking Sequence Algorithm}
    \label{alg:WPS-mechanism}
    \begin{algorithmic}[1]
        \REQUIRE $N,M,(u_i(g))_{i\in N, g\in M}, (w_i)_{i\in N}$
        \STATE $A_i \leftarrow \emptyset$ and $t_i \leftarrow 0$ for all $i \in N$
        \STATE $M_0 \leftarrow M$
        \WHILE{$M_0 \neq \emptyset$} \label{code:for-start}
            \STATE $i^* \leftarrow \mathrm{argmin}_{i\in N} \frac{t_i}{w_i}$ 
            \STATE $g^* \leftarrow \mathrm{argmax}_{g\in M_0} u_{i^*}(g)$
            \STATE $A_{i^*} \leftarrow A_{i^*} \cup \{g^*\}$
            \STATE $M_0 \leftarrow M_0 \setminus \{g^*\}$
            \STATE $t_{i^*} \leftarrow t_{i^*} + 1$
        \ENDWHILE \label{code:wps}
        \RETURN $(A_1,A_2,\ldots,A_n)$
    \end{algorithmic}
\end{algorithm}

Denote by $t_i(s)$ the number of times agent $i$ has picked an item up to (and including) step $s$.
For each $i\in N$ and each $k \in [t_i(m)]$, let $s^i(k)$ be the step where agent $i$ picks her $k$-th item, and denote this item by $g_{k}^{i}$; for convenience, let $s^i(0) = 0$.
Finally, let $A=(A_1,A_2,\ldots,A_n)$ be the allocation returned by Algorithm~\ref{alg:WPS-mechanism}.

\citet[Thm.~3.3]{ChakrabortyIgSu21} gave a rather long algebraic proof that Algorithm~\ref{alg:WPS-mechanism} always returns a WEF1 allocation.
\citet[Lem.~A.2]{WuZhZh23} provided an alternative proof involving integrals.
We present a relatively succinct algebraic proof of this statement via Abel's summation formula.

\begin{proposition}[\citep{ChakrabortyIgSu21}]
\label{prop:WEF1-alternative}
The allocation returned by Algorithm~\ref{alg:WPS-mechanism} is WEF1.
\end{proposition}

\begin{proof}
Consider any two distinct agents $i,j \in M$; it suffices to show that $i$ is WEF1 towards~$j$.
Observe that each of $i$ and $j$ picks her first item before the other agent picks her second item.
Let $\tau_1+1$ be the number of items that $j$ picks before step $s^i(2)$, and denote these items by $g_1^j, g_{1,1}^{j}, g_{1,2}^{j}, \ldots, g_{1,\tau_1}^{j}$. 
Note that $g_1^j$ may be picked either before or after $g_1^i$. 

For each $k \in [t_i(m)-1]\setminus\{1\}$, let $\tau_k$ be the number of items that agent $j$ picks between steps $s^i(k)$ and $s^i(k+1)$, and for $\ell\in[\tau_k]$, let $g_{k,\ell}^{j}$ be the $\ell$-th item picked by agent $j$ within this interval.
Let $\tau_{t_i(m)}$ denote the number of items that $j$ picks after step $s^i(t_i(m))$, and for $\ell \in [\tau_{t_i(m)}]$, let $g_{t_i(m),\ell}^{j}$ be the $\ell$-th item picked by $j$ within this interval.
Note that $\tau_k$ may be zero for some $k$.
See Figure~\ref{fig:item-order} for an illustration.

We claim that for every $k \in [t_i(m)]$,  
\begin{equation}\label{eq:upper_bound_upsilon_WEF1}
    \frac{\sum_{k'=1}^{k} \tau_{k'}}{w_j} \le \frac{k}{w_i}.
\end{equation}
Indeed, if $\tau_k \ge 1$, then \eqref{eq:upper_bound_upsilon_WEF1} holds since agent $j$ is allowed to pick item $g_{k,\tau_k}^j$.
If $\tau_k = 0$, for $k=1$ we have $\frac{\sum_{k'=1}^1 \tau_{k'}}{w_j} = 0 \le \frac{1}{w_i}$, while for $k \ge 2$ we have $\frac{\sum_{k'=1}^{k-1} \tau_{k'}}{w_j} \le \frac{k-1}{w_i}$ from the property for $k-1$, which implies that $\frac{\sum_{k'=1}^k \tau_{k'}}{w_j} = \frac{\sum_{k'=1}^{k-1} \tau_{k'}}{w_j} \le \frac{k-1}{w_i} < \frac{k}{w_i}$. 
Thus, \eqref{eq:upper_bound_upsilon_WEF1} holds for all $k \in [t_i(m)]$.

For each $k \in [t_i(m)]$, let $\eta_k = 1 - \frac{w_i}{w_j} \sum_{\ell=1}^{\tau_{k}} \frac{u_i(g_{k,\ell}^{j})}{u_i(g_{k}^{i})}$ when $\tau_k \ge 1$, and $\eta_k = 1$ when $\tau_k = 0$.
Since $u_i(g_{k}^{i}) \ge u_i(g_{k,\ell}^{j})$ for all $\ell\in[\tau_k]$, we have $\eta_k \ge 1 - \frac{w_i}{w_j}\cdot\tau_{k}$ for all $k$.
Hence, \eqref{eq:upper_bound_upsilon_WEF1} implies that
$\sum_{k'=1}^{k} \eta_{k'} 
    \ge 
    k - \frac{w_i}{w_j} \sum_{k'=1}^{k} \tau_{k'} \ge 0$ for all $k \in [t_i(m)]$.
This means that
\begin{align*}\label{eq:proof_of_WEF1}
    &u_i(A_i) - \frac{w_i}{w_j}\cdot u_i(A_j \setminus \{g_{1}^{j}\}) \nonumber\\
    &=  
        \sum_{k=1}^{t_i(m)} u_i(g_{k}^{i})  - \frac{w_i}{w_j}\sum_{k=1}^{t_i(m)} \sum_{\ell=1}^{\tau_{k}} u_i(g_{k,\ell}^{j}) 
        =  
        \sum_{k=1}^{t_i(m)} u_i(g_{k}^{i}) \cdot \eta_k \nonumber \\
        &= u_i(g_{t_i(m)}^{i}) \sum_{k=1}^{t_i(m)} \eta_k \\
        &\quad+ \sum_{k=1}^{t_i(m)-1} \left[\left(u_i(g_{k}^{i}) - u_i(g_{k+1}^{i}) \right)
    \sum_{k'=1}^{k} \eta_{k'}\right] \ge 0,
\end{align*}
where the last equality follows from Lemma~\ref{lemma:Abel_summation} and the inequality from the fact that $u_i(g^i_k) \ge u_i(g^i_{k+1})$ for every $k$.
Hence, agent~$i$ is WEF1 towards agent~$j$, as desired. 
\end{proof}

\subsection{Asymptotic Result}
We now present the main result of this section.
\begin{theorem}\label{thm:WeightedPicking}
    Suppose that $\mathcal{D}$ is PDF-bounded, and let $C\ge 1$ be an arbitrary constant.
    For any weight vector $(w_1,w_2,\ldots , w_n)$ such that $w_{\max}/w_{\min} \le C$, if $m = \Omega\left(n \log n /\log \log n \right)$, then Algorithm~\ref{alg:WPS-mechanism} produces a WEF allocation with high probability. 
\end{theorem}

Before we prove \Cref{thm:WeightedPicking}, we introduce some notation.
For any $c \in (0, 1]$, denote by $\mathcal{D}_{\le c}$ the conditional distribution of $\mathcal{D}$ on $[0,c]$.
For any positive integer $k$, denote by $\mathcal{D}^{\max(k)}$ the distribution of the maximum of $k$ independent random variables generated by $\mathcal{D}$.

Fix any two distinct agents $i,j\in N$.
We will use the same notation as in \Cref{sec:WEF1-alternative}.
Let $X_{1}^{j} = u_i(g_{1}^{j})$.
For each $k \in [t_i(m)]$ and $\ell \in [\tau_k]$,
let $X_{k} = u_i(g_{k}^i)$ and $X_{k,\ell}^{j} = u_i(g_{k,\ell}^{j})$.
Lemma A.1 of \citet{ManurangsiSu21} yields the following lemma, which gives a convenient description of the distributions of $X_k$ and $X_{k, \ell}^j$.

\begin{lemma}[\citep{ManurangsiSu21}]
\label{lemma:picking-generation}
    Let $X_0 = 1$. 
    Then, $(X_k)_{k\in [t_i(m)]}$ and $(X_{k,\ell}^{j})_{k\in [t_i(m)],\, \ell \in [\tau_k]}$ can be generated according to the following process:
    \begin{itemize}
        \item For each $k \in[t_i(m)]$, let
        $X_{k} \sim \mathcal{D}_{\le X_{k-1}}^{\max(m - s^i(k) + 1)}$;
        \begin{itemize}
            \item For each $\ell \in[\tau_k]$, let
            $X_{k,\ell}^{j} \sim \mathcal{D}_{\le X_{k}}$.
        \end{itemize}
    \end{itemize}
\end{lemma}

Intuitively, before agent~$i$ picks the item $g_k^i$ corresponding to $X_k$, there are $m - s^i(k)+1$ items remaining, and $i$'s utility for each of them cannot exceed her utility for the item $g_{k-1}^i$ corresponding to $X_{k-1}$.
Moreover, $i$'s utility for each item $g^j_{k,\ell}$ corresponding to $X^j_{k,\ell}$ picked by $j$ cannot exceed her utility for $g_k^i$, but otherwise $j$'s picks are independent of $i$'s valuations.

We now proceed to the proof of Theorem~\ref{thm:WeightedPicking}.
\begin{proof}[Proof of Theorem~\ref{thm:WeightedPicking}]
Suppose that $\mathcal{D}$ is $(\alpha,\beta)$-PDF-bounded and $m \ge 10^6 \tilde{\beta} \cdot C \cdot n \log n / \log \log n$, where $\tilde{\beta} := \beta/\alpha \ge 1$. 

Recall that for any $k\in [t_i(m)]$, we defined 
\begin{align*}
\eta_k = 1 - \frac{w_i}{w_j} \sum_{\ell=1}^{\tau_{k}} \frac{u_i(g_{k,\ell}^{j})}{u_i(g_{k}^{i})} = 1 - \frac{w_i}{w_j} \sum_{\ell=1}^{\tau_{k}} \frac{X_{k,\ell}^{j}}{X_{k}}.
\end{align*}
For each $k$, let 
$
\theta_k = 1 - \frac{w_i}{w_j} \cdot \tau_{k}.
$
By description of the algorithm, $X_k \ge X_{k,\ell}^{j}$ holds for all $k\in [t_i(m)]$ and $\ell\in[\tau_k]$, so $\eta_k \ge \theta_k$ for all $k$. 
From~\eqref{eq:upper_bound_upsilon_WEF1}, we obtain that for any $k$,
\begin{equation}\label{eq:cumulative_sum_T}
    \sum_{k'=1}^{k} \theta_{k'} = k - \frac{w_i}{w_j} \sum_{k'=1}^{k} \tau_{k'} \ge 0.
\end{equation} 
By Lemma~\ref{lemma:Abel_summation} combined with \eqref{eq:cumulative_sum_T} and the fact that $X_{k} \ge X_{k+1}$ for every $k$, we have 
\begin{align*}
&\sum_{k=1}^{t_i(m)} X_k \cdot \theta_k \\
&= X_{t_i(m)} \sum_{k=1}^{t_i(m)} \theta_k  + \sum_{k=1}^{t_i(m)-1} \left(X_{k} - X_{k+1} \right) \sum_{k'=1}^{k} \theta_{k'}\ge 0.
\end{align*}
From this, we get
\begin{align*}
    u_i(A_i) - \frac{w_i}{w_j}\cdot u_i(A_j) 
    &=
        \sum_{k=1}^{t_i(m)} X_{k}\cdot \eta_k - \frac{w_i}{w_j}\cdot X_1^j \\
    &\ge 
        \sum_{k=1}^{t_i(m)} X_{k}\cdot \eta_k - \frac{w_i}{w_j} \tag{by $X_1^j\le 1$}\\
    &\ge 
        \sum_{k=1}^{t_i(m)} X_{k} \left(\eta_k - \theta_k\right)  - \frac{w_i}{w_j}.
\end{align*}

Our goal is to demonstrate that, with high probability, $\sum_{k=1}^{t_i(m)} X_{k} \left(\eta_k - \theta_k\right)  - \frac{w_i}{w_j} \ge 0$ holds.
To this end, we show that with high probability, there exists a point in the picking sequence where agent $i$ receives an item of sufficiently high value and has accumulated enough advantage over agent $j$.
This is formalized in the following lemma.

\begin{lemma}\label{lem:existence_of_T}
    Suppose that $\mathcal{D}$ is $(\alpha,\beta)$-PDF-bounded, and $m \geq 10^6 \tilde{\beta} \cdot C \cdot n \log n/\log \log n$. 
    For any distinct agents $i,j\in N$, there exists a positive integer $T$ such that 
    \begin{enumerate}
        \item[$(a)$] 
            $t_i(m) \ge T \ge 1$;
        \item[$(b)$] 
            $\mathrm{Pr}\left[X_T \ge \frac{1}{2} \right] = 1-O\left(\frac{1}{m^3}\right)$; and
        \item[$(c)$]      
            $\mathrm{Pr}\left[\sum_{k=1}^{T}  \sum_{\ell=1}^{\tau_{k}}\left(1-  \frac{X_{k,\ell}^{j}}{X_{k}} \right)\ge 2 \right] = 1-O\left(\frac{1}{m^3}\right)$.
    \end{enumerate} 
\end{lemma}

We defer the proof of Lemma~\ref{lem:existence_of_T} to Appendix~\ref{appendix:proof_existence_of_T}.
Using this lemma, we will show that with probability at least $1-O\left(1/m^3\right)$, agent $i$ does not envy agent $j$.
We have
\begin{align*}
    u_i(A_i) - \frac{w_i}{w_j}u_i(A_j)
    &\ge 
        \sum_{k=1}^{t_i(m)} X_{k} \left(\eta_k - \theta_k\right)  - \frac{w_i}{w_j} \\
    &\ge 
        \sum_{k=1}^{T} X_{k}\left(\eta_k - \theta_k\right)  - \frac{w_i}{w_j} 
        \tag{by $X_{k}\ge 0$ and $\eta_k \ge \theta_k $ for any $k$, and (a) in Lemma~\ref{lem:existence_of_T}}\\
    &\ge 
        X_{T}\sum_{k=1}^{T}  \left(\eta_k - \theta_k\right)  - \frac{w_i}{w_j} 
        \tag{by $X_1\ge X_2\ge \cdots \ge X_T$}\\
    &= 
        X_T \frac{w_i}{w_j} \sum_{k=1}^{T}  \sum_{\ell=1}^{\tau_{k}}\left(1-  \frac{X_{k,\ell}^{j}}{X_{k}} \right) - \frac{w_i}{w_j} \\
    &\ge 
        \frac{w_i}{2 w_j} \sum_{k=1}^{T}  \sum_{\ell=1}^{\tau_{k}}\left(1-  \frac{X_{k,\ell}^{j}}{X_{k}} \right) - \frac{w_i}{w_j} \tag{with probability at least $1-O\left(1/m^3\right)$, by (b) in Lemma~\ref{lem:existence_of_T}}\\
    &\ge 
        \frac{w_i}{2 w_j}\cdot  2 - \frac{w_i}{w_j} \tag{with probability at least $1-O\left(1/m^3\right)$, by (c) in Lemma~\ref{lem:existence_of_T}}\\
    &=
        0.
\end{align*}
Finally, by the union bound over all pairs of agents $i,j$ and the fact that $m\ge n$, we have that the probability that the allocation~$A$ is \emph{not} WEF is at most $n^2\cdot O(1/m^3) = O(1/n)$.
This completes the proof of the theorem.
\end{proof}

\section{Weighted Proportionality}
\label{sec:WPROP}

In this section, we turn our attention to weighted proportionality, and establish a sharp threshold for the transition from non-existence to existence with respect to this notion.

We start with the non-existence.

\begin{theorem}\label{thm:WPROP_non_existence}
    Suppose that $\mathcal{D}$ is a non-atomic distribution with mean $\mu\in (0,1)$, and let $\varepsilon \in (0, 1/2)$ be a  constant. 
    For any $n$, 
    there exists a weight vector $(w_1,w_2,\ldots , w_n)$ with $\frac{w_{\max}}{w_{\min}} \le \frac{2}{(1-\mu)^2\varepsilon^2}$ such that for any $m \le (1 - \varepsilon) \cdot \frac{n}{1-\mu} $, with high probability, no WPROP allocation exists.
\end{theorem}

\begin{proof}
    When $n > m$, at least one agent receives no item. 
    Since $\mathcal{D}$ is non-atomic, with probability $1$, all agents have positive utilities for all items, in which case no WPROP allocation exists.
    We thus assume that $n \le m \le (1 - \varepsilon) \cdot \frac{n}{1-\mu}$.

    Let $\delta = (1-\mu)\varepsilon < 1/2$.
    For sufficiently large $n$, we have $\lfloor\delta n \rfloor\ge 2$.
    Assume that there are $\lceil(1-\delta) n\rceil$ agents with weight $ \delta  / \lceil(1-\delta) n\rceil$ each (called agents of the ``first type''), and the remaining $\lfloor\delta n \rfloor$ agents have weight $ (1-\delta)  / \lfloor\delta n \rfloor$ each (called agents of the ``second type'').
    Note that the sum of all weights is $\delta + (1-\delta) = 1$, and the ratio between the maximum and minimum weights is
    \begin{align*}
    \frac{w_{\max}}{w_{\min}} 
    &= \frac{1-\delta}{\lfloor\delta n \rfloor}\cdot \frac{\lceil(1-\delta) n\rceil}{\delta} \\
    &\le \frac{1}{\delta n / 2} \cdot \frac{n}{\delta} = \frac{2}{\delta^2} = \frac{2}{(1-\mu)^2 \varepsilon^2}.
    \end{align*}
    Our choice of $\delta$ implies that $(1 - \delta)(1 - \mu \eps) > 1 - \delta - \mu\eps = 1 - \eps$. Rearranging this yields $\frac{1 - \mu \eps}{1 
- \eps} - \frac{1}{1 - \delta} > 0$. Let $\gamma = \frac{1}{2}\left(\frac{1 - \mu \eps}{1 
- \eps} - \frac{1}{1 - \delta}\right)$; by the previous sentence, $\gamma > 0$.

    Note that $\mathbb{E}\left[ u_i(M) \right] = m \mu$ for each $i\in N$. 
    Applying \Cref{lemma:Chernoff}, we get 
    $$
    \mathrm{Pr}\left[u_i(M) \le (1-\gamma) m \mu \right] \le \mathrm{exp}\left({-\frac{{\gamma}^2 m \mu}{2}}\right).
    $$
    By the union bound, the probability that there exists $i\in N$ with $u_i(M) \le (1-\gamma) m \mu$ is at most $n \cdot \mathrm{exp}\left({-{\gamma}^2 m \mu / 2}\right)$.
    Since $m\ge n$,
    we have $n \cdot \mathrm{exp}\left({-{\gamma}^2 m \mu / 2}\right) \le n \cdot \mathrm{exp}\left({-{\gamma}^2 n \mu / 2}\right)$, which approaches $0$ as $n\to \infty$.
    Thus, with high probability, $u_i(M) \ge (1-\gamma) m \mu$ holds for all $i\in N$. 
    We assume for the remainder of the proof that this event holds.

    For an allocation to satisfy WPROP, each agent of the first type must receive at least one item.
    Each agent of the second type, given her weight of $\frac{1-\delta}{\lfloor\delta n \rfloor}$, must receive a bundle that she values at least $\frac{1-\delta}{\lfloor\delta n \rfloor} \cdot (1-\gamma) m \mu $.
    Since $u_i(g) \le 1$ for all $i\in N$ and $g\in M$, for each agent $i$ of the second type, her bundle $A_{i}$ must contain at least $ \frac{1-\delta}{\lfloor\delta n \rfloor} \cdot (1-\gamma) m \mu$ items.
    Therefore, for a WPROP allocation to exist, the total number of items must be at least $\lceil(1-\delta)n\rceil + \lfloor\delta n \rfloor \cdot \frac{1-\delta}{\lfloor\delta n \rfloor} \cdot (1-\gamma) m \mu \ge (1-\delta) \left(n  + (1-\gamma) m \mu \right)$.
    Since $m \le \frac{1 - \varepsilon}{1-\mu} \cdot n$, this is at least 
    \begin{align*}
    & (1-\delta)\left(\frac{1-\mu}{1-\varepsilon}  + (1-\gamma)  \mu\right) m \\
    &\geq (1-\delta)\left(\frac{1-\mu}{1-\varepsilon}  + \mu - \gamma\right) m \\
    &= (1-\delta)\left(\gamma + \frac{1}{1 - \delta}\right) m 
    > m,
    \end{align*}
    a contradiction. 
    We therefore conclude that, with high probability, no WPROP allocation exists.
\end{proof}

Next, we exhibit the asymptotic existence of weighted proportional allocations.

\begin{theorem}
\label{thm:WPROP-existence}
    Suppose that $\mathcal{D}$ is a PDF-bounded distribution with mean $\mu\in (0,1)$, and let $C \ge 1$ and $\varepsilon \in (0, 1)$ be arbitrary constants. 
    For any weight vector $(w_1,w_2,\ldots , w_n)$ with $\frac{w_{\max}}{w_{\min}} \le C$ and any $m \ge (1 + \varepsilon)\cdot \frac{n}{1-\mu}$, a WPROP allocation exists with high probability.
    Moreover, such an allocation can be found in polynomial time.
\end{theorem}

    As WEF is stronger than WPROP, Theorem~\ref{thm:WeightedPicking} already implies that a WPROP allocation exists with high probability when $m =\Omega\left( n \log n /  \log \log n \right)$. 
    Hence, to prove \Cref{thm:WPROP-existence}, we may restrict our attention to the case $m=O(n \log n)$. 
    Let $\cD$ be $(\alpha,\beta)$-PDF-bounded, and let $\tau = 1- \frac{8(C + 1) \log m}{\alpha n}$ and $\delta = \left(1+\frac{\varepsilon}{1 + \varepsilon}\cdot \frac{1-\mu}{\mu} \right)\tau - 1$. Note that $\tau>0$ and $\delta > 0$ for sufficiently large $n$.

    Since $\mathbb{E}\left[ u_i(M) \right] = m \mu$ for all $i\in N$, we can apply \Cref{lemma:Chernoff}  to get that
    $
        \mathrm{Pr}\left[u_i(M) \ge (1+\delta) m \mu \right] \le \mathrm{exp}\left({-\frac{{\delta}^2 m \mu}{2+\delta}}\right).
    $
    By the union bound and the assumption that $m\ge (1 + \varepsilon)\cdot \frac{n}{1-\mu} \ge n$, with high probability, we have $u_i(M) \le (1+\delta) m \mu$ for all agents $i\in N$.
    We assume for the remainder of this discussion that this event holds.

    Recall that $W = \sum_{i\in N} w_i$.
    For each $i\in N$, define $s_i = \left\lceil(1 + \delta) \frac{w_i}{W}  \cdot \frac{\mu m}{\tau}\right\rceil$.
    If we can construct an allocation $A$ where every agent $i$ receives at least $s_i$ items that she values at least $\tau$ each, then  with high probability,
    $    u_i(A_i) \ge s_i \cdot \tau
        \ge  \frac{w_i}{W}\cdot (1 + \delta)m \mu 
        \ge \frac{w_i}{W}\cdot u_i(M)$
    for all $i\in N$, implying that the allocation $A$ satisfies WPROP. 
    
    Let $\bm{s}=(s_1,s_2,\ldots, s_n)$. As $m \ge (1+\varepsilon)\cdot\frac{n}{1-\mu}$, we have
    \begin{align*}
        \sum_{i\in N}  s_i 
        &\le  
            \sum_{i\in N} (1 + \delta) \frac{w_i}{W} \cdot \frac{\mu m}{\tau} + n \\
        &= 
            (1 + \delta) \cdot \frac{\mu m}{\tau} + n \\
        &\le 
            (1 + \delta) \cdot \frac{\mu m}{\tau} + \frac{1-\mu}{1 + \varepsilon} \cdot m \\
        &= 
            \left(1 + \frac{\varepsilon}{1 + \varepsilon} \cdot \frac{1-\mu}{\mu} \right)  \mu m + \frac{1-\mu}{1 + \varepsilon} \cdot m 
        = m.
    \end{align*}
    This calculation shows that the total number of required items does not exceed the number of available items. 
    
    To find a desired allocation, we use a matching-based approach, which extends an algorithm of \citet{ManurangsiSu20} to the weighted setting. 
    Before describing the algorithm, we define a generalized notion of a matching that allows vertices on one side to be matched multiple times.
     \begin{definition}
        Let $G=(L\cup R, E)$ be a bipartite graph with $|L| = n$ and $|R|=m$. 
        For a vector of positive integers $\bm{s}=(s_1,s_2,\ldots, s_n)$ with $\sum_{i=1}^n s_i \le m$, an $\bm{s}$-\emph{matching} in $G$ is a set of edges $F\subseteq E$ such that each vertex $i\in L$ is incident to at most $s_i$ edges in $F$ and each vertex in $R$ is incident to at most one edge in $F$.
        An $\bm{s}$-matching is called \emph{left-saturating} if every vertex $i \in L$ is incident to exactly $s_i$ edges in the matching.
    \end{definition}
    
\begin{algorithm}[h]
    \caption{Matching-Based Algorithm for WPROP}
    \label{alg:matching-based-algorithm}
    \begin{algorithmic}[1]
        \REQUIRE $N,M,(u_i(g))_{i\in N, g\in M}$, $\bm{s}=(s_i)_{i\in N}$, threshold $\tau$
        \FOR{$i\in N$}
            \STATE $M_{\ge \tau}(i) \leftarrow \{j\in M \mid u_i(j)\ge \tau\}$
        \ENDFOR
        \STATE Let $G_{\ge \tau} = (N, M, E_{\ge \tau})$ be the bipartite graph where $(i, j) \in E_{\ge \tau}$ if and only if $j \in M_{\ge \tau}(i)$.
        \IF{$G_{\ge \tau}$ contains a left-saturating $\bm{s}$-matching}
            \RETURN  any left-saturating $\bm{s}$-matching in $G_{\ge \tau}$ (with the unmatched items allocated arbitrarily)
        \ELSE
        \RETURN NULL
        \ENDIF
    \end{algorithmic}
\end{algorithm}

    Our algorithm, described as Algorithm~\ref{alg:matching-based-algorithm}, first constructs a bipartite graph $G_{\ge \tau}$ where an edge exists between an agent and an item if and only if the agent values the item at least the threshold $\tau$. 
    The algorithm then determines whether a left-saturating $\bm{s}$-matching exists in this graph.
    If so, it simply assigns each matched item to the agent matched with it, and any unmatched item can be assigned arbitrarily. 
    Note that determining whether a left-saturating $\bm{s}$-matching exists, and finding one in case it does, can be done in polynomial time by creating $s_i$ copies of each agent~$i$ and computing a maximum matching in the resulting graph.

    We now state a lemma that establishes the existence of a left-saturating $\bm{s}$-matching. 
    \begin{lemma}\label{lemma:existence_of_perfect_s_matching}
        Suppose that $\mathcal{D}$ is $(\alpha,\beta)$-PDF-bounded,
        there exists a constant $C \ge 1$ such that $\frac{w_{\max}}{w_{\min}} \le C$, and $m=O(n\log n)$.
        Set $\tau = 1- \frac{8 (C + 1) \log m}{\alpha n}$ in Algorithm~\ref{alg:matching-based-algorithm}. Then, with high probability, Algorithm~\ref{alg:matching-based-algorithm} outputs a left-saturating $\bm{s}$-matching.
    \end{lemma}
    
    From our earlier discussion,  \Cref{lemma:existence_of_perfect_s_matching} implies that a WPROP allocation exists (and can be found in polynomial time) with high probability, thereby yielding \Cref{thm:WPROP-existence}.

To prove Lemma~\ref{lemma:existence_of_perfect_s_matching}, we first recall basic results from matching theory.
For a bipartite graph $G=(L\cup R, E)$ and a subset $Y\subseteq L$, denote by $N_G(Y)$ the set of vertices in $R$ that are adjacent to at least one vertex in $Y$. 
A \emph{matching} in a graph is a set of edges no two of which share a vertex. 
A matching is called \emph{left-saturating} if every vertex in $L$ is incident to exactly one edge in the matching. 
We now state Hall's marriage theorem, a classical result in matching theory.
\begin{lemma}[Hall’s marriage theorem]
    Let $G=(L\cup R, E)$ be a bipartite graph with $|L| \le |R|$. If $|N_G(Y)| \ge |Y|$ for all subsets $Y\subseteq L$, there exists a left-saturating matching in~$G$.
\end{lemma}
To extend this result to our setting with $\bm{s}$-matchings, we construct a new bipartite graph $G'$ from $G=(L\cup R, E)$ by replacing each vertex $i \in L$ with $s_i$ copies, where each copy is connected to all neighbors of $i$ in $G$. 
This immediately leads to the following proposition.
\begin{proposition}\label{proposition:a_perfect_s-matching}
    Let $G=(L\cup R, E)$ be a bipartite graph with $n = |L| \le |R| = m$.
    For any vector of positive integers $\bm{s}=(s_1,s_2,\ldots, s_n)$ such that $\sum_{i=1}^n s_i \le m$, if $|N_G(Y)| \ge \sum_{i\in Y} s_i$ for all subsets $Y\subseteq L$, then there exists a left-saturating $\bm{s}$-matching in $G$.
\end{proposition}

Next, we recall the \emph{Erd\H{o}s-R\'{e}nyi random bipartite graph model}~\citep{ErdosRe64}.
For any $p \in [0, 1]$, 
let $\mathcal{G}(|L|, |R|, p)$ denote the probability distribution over bipartite graphs with vertex sets $L$ and $R$ such that for each pair of vertices $i \in L$ and $j \in R$, the edge $(i,j)$ exists independently with probability $p$.
Let $s_{\max} = \max_{i\in N}s_i$ and $s_{\min} = \min_{i\in N}s_i$.
The probability that no left-saturating $\bm{s}$-matching exists can be upper-bounded by the probability that there exists a subset violating the condition in Proposition~\ref{proposition:a_perfect_s-matching}. Based on this insight, we establish the following lemma, whose proof is deferred to Appendix~\ref{appendix:proof_existence_of_perfect_s_matching_in_random_graph}.
\begin{lemma}\label{lemma:existence_of_perfect_s_matching_in_random_graph}
    Let $B\ge 1$ be a constant, and
    let $G$ be a bipartite graph sampled from $\mathcal{G}(n, m, p)$ with $n\le m$ and $\frac{8 B \log m}{n}\le p \le 1$.
    For any vector of positive integers $\bm{s}=(s_1,s_2,\ldots, s_n)$ such that $\sum_{i=1}^n s_i \le m$ and $\frac{s_{\max}}{s_{\min}} \le B$, with high probability, $G$ contains a left-saturating $\bm{s}$-matching.
\end{lemma}

We can now complete the proof of Lemma~\ref{lemma:existence_of_perfect_s_matching}.
\begin{proof}[Proof of Lemma~\ref{lemma:existence_of_perfect_s_matching}]
To apply Lemma~\ref{lemma:existence_of_perfect_s_matching_in_random_graph},
we verify two conditions: firstly, that $\frac{s_{\max}}{s_{\min}} \le C + 1$, and secondly, that the graph $G_{\ge \tau}$ follows the Erd\H{o}s-R\'{e}nyi random bipartite graph model. 
Recall that $s_i = \left\lceil(1 + \delta)  \frac{w_i}{W} \cdot \frac{\mu m}{\tau}\right\rceil$.  
Since $s_{\min} \geq 1$ and $\frac{w_{\max}}{w_{\min}} \le C$, we have
    \begin{align*}
        \frac{s_{\max}}{s_{\min}} 
        &\le 
            \frac{(1 + \delta)  \frac{w_{\max}}{W}  \cdot \frac{\mu m}{\tau} + 1}{s_{\min}} \\
        &\le 
            \frac{(1 + \delta)  \frac{w_{\max}}{W} \cdot \frac{\mu m}{\tau} }{(1 + \delta) \frac{w_{\min}}{W} \cdot \frac{\mu m}{\tau}} + 1
        \le 
            C + 1.
    \end{align*}
    Moreover, since $\mathcal{D}$ is $(\alpha,\beta)$-PDF-bounded, the events $\{u_i(j) \ge \tau\}$ occur independently for all pairs $(i,j) \in N \times M$, each with probability
    \begin{align*}
    \mathrm{Pr}\left[ u_i(j) \ge \tau \right] 
    &= \mathrm{Pr}\left[ u_i(j) \ge 1-\frac{8(C + 1) \log m}{\alpha  n} \right] \\
    &\ge \alpha \cdot \frac{8 (C + 1) \log m}{\alpha n} = \frac{8 (C + 1) \log m}{n}.
    \end{align*}
    Therefore, by setting $B= C + 1$ in Lemma~\ref{lemma:existence_of_perfect_s_matching_in_random_graph}, we conclude that the graph $G_{\ge \tau}$ contains a left-saturating $\bm{s}$-matching with high probability, as desired.
\end{proof}

\section{Two Agents}
\label{sec:two-agents}

In this section, we focus on the case of $n = 2$ agents---note that WEF and WPROP are equivalent in this case.
We let $r = w_2/w_1 \ge 1$ and consider the asymptotics as $r\to\infty$.
Our main result is that if the number of items grows asymptotically faster than $\sqrt{r}$, a WEF allocation is likely to exist.

\begin{theorem}\label{thm:two_agents}
    Suppose that there are two agents with weights $(w_1,w_2)$ where $r := w_2/w_1 \geq 1$, and that $\mathcal{D}$ is PDF-bounded. 
    If $m = \omega(\sqrt{r})$, then a WEF allocation exists with probability approaching $1$ as $r \to \infty$, and such an allocation can be found in polynomial time. 
    The same holds for WPROP.
\end{theorem}
\begin{proof}
Since WEF and WPROP are equivalent for $n = 2$, it suffices to focus on WPROP.
Suppose that $\mathcal{D}$ is $(\alpha,\beta)$-PDF-bounded with mean $\mu \in (0,1)$.

We present an allocation algorithm.
Set $p = \frac{\alpha \mu}{2\sqrt{r+1}} \in (0,1)$, and define $\tau$ as the threshold such that $\Pr[X \leq \tau] = p$ for $X \sim \mathcal{D}$.
The algorithm allocates each item whose value to agent $2$ is less than $\tau$ to agent $1$, and the remaining items to agent $2$.
Clearly, this algorithm runs in polynomial time.
Let $A = (A_1,A_2)$ denote the resulting allocation.
To show that $A$ is WPROP with probability $1 - o(1)$, it suffices to prove that each of the two inequalities $(r+1)u_1(A_1) < u_1(M)$ and $(r+1)u_2(A_1) > u_2(M)$ holds with probability $o(1)$.

First, consider agent $1$.
For any sufficiently large $r$, we have $p(r + 1) \geq 3$. When this holds, the union bound yields
\begin{align*}
&\Pr[(r+1)u_1(A_1) < u_1(M)] \\
&\leq \Pr\left[u_1(M) > \frac{3}{2} m\mu\right] + \Pr\left[u_1(A_1) < \frac{1}{2} mp\mu\right].
\end{align*}
We may bound each term above via \Cref{lemma:Chernoff}. 
In particular, $u_1(M)$ is a sum of $m$ independent random variables drawn from $\cD$, so $\E[u_1(M)] = m\mu$; \Cref{lemma:Chernoff} then implies that $\Pr[u_1(M) > \frac{3}{2} m\mu] \leq \mathrm{exp}\left( - \frac{1}{10} m \mu \right) = o(1)$. 
Meanwhile, since each item is independently allocated to agent $1$ with probability $p$, $u_1(A_1)$ is a sum of independent random variables with $\E[u_1(A_1)] = mp\mu$, so \Cref{lemma:Chernoff} implies that $\Pr\left[ u_1(A_1) < \frac{1}{2}  mp\mu  \right] 
\le \mathrm{exp}\left(- \frac{1}{8}mp \mu \right)$, which is $o(1)$ since $m = \omega(\sqrt{r})$. 
Combining these, we get $\Pr[(r+1)u_1(A_1) < u_1(M)] = o(1)$, as desired.

Next, consider agent $2$.
By the union bound, we have
\begin{align*}
&\Pr[(r+1)u_2(A_1) > u_2(M)] \\
&\leq \Pr\left[u_2(M) < \frac{1}{2} m\mu\right] + \Pr\left[u_2(A_1) > \frac{m \mu}{2(r+1)}\right].
\end{align*}
For the first term, we can again use \Cref{lemma:Chernoff} in a similar way as above to conclude that $\Pr\left[u_2(M) < \frac{1}{2} m\mu\right] \le \mathrm{exp} \left( - \frac{1}{8} m \mu \right) = o(1)$. 

We now focus on the second term $\Pr\left[u_2(A_1) > \frac{m \mu}{2(r+1)}\right]$. 
For every item $g \in M$, define $X_g$ to be a random variable such that $X_g=1$ if $g\in A_1$, and $X_g=0$ otherwise. 
Note that each $X_g$ is an independent random variable with $\E[X_g] = p$.
Since $\mathcal{D}$ is $(\alpha,\beta)$-PDF-bounded, we have $\alpha\tau \le \mathrm{Pr}[X\le \tau] = p$, and so $\tau \le \frac{\mu}{2 \sqrt{r+1}}$.
By definition of the allocation, every item $g \in A_1$ satisfies $u_2(g) \le \tau$. 
Therefore, we have
$u_2(A_1) = \sum_{g\in A_1} u_2(g) \le 
\tau \cdot \sum_{g\in M}  X_g \le \frac{\mu}{2\sqrt{r+1}} \sum_{g\in M}  X_g$.
Observe that $\frac{3}{2}\mathbb{E}\left[\sum_{g\in M} X_g\right] = \frac{3}{2} mp = \frac{3  m \alpha \mu}{4\sqrt{r+1}} \le \frac{m}{\sqrt{r+1}}$ since $\alpha \le 1$ and $\mu<1$. Using this inequality and the bound on $u_2(A_1)$, we obtain
\begin{align*}
    &\mathrm{Pr}\left[ u_2(A_1) > \frac{m \mu}{2(r+1)} \right] 
    \le
        \mathrm{Pr}\left[ \sum_{g\in M}  X_g > \frac{m}{\sqrt{r+1}} \right] \\
    &\le
        \mathrm{Pr}\left[ \sum_{g\in M}  X_g > \frac{3}{2}\mathbb{E}\left[\sum_{g\in M}  X_g\right] \right] \\
    &\le 
        \mathrm{exp} \left( -\frac{1}{10}\mathbb{E}\left[\sum_{g\in M}  X_g\right]  \right) 
        = 
        \mathrm{exp} \left( -\frac{1}{10} m p  \right),
\end{align*}
where we use \Cref{lemma:Chernoff} for the last inequality. 
This probability is $o(1)$ since $m = \omega(\sqrt{r})$.

Hence, combining our bounds for both agents, we conclude that $A$ is WPROP with probability $1 - o(1)$.
\end{proof}

Next, we establish the tightness of the bound in \Cref{thm:two_agents}.
In particular, the following theorem implies that the probability that \emph{no} WEF allocation exists is at least constant if $m = O(\sqrt{r})$, and approaches $1$ if $m = o(\sqrt{r})$.

\begin{theorem}\label{thm:two_agents_nonexistence}
    Suppose that there are two agents with weights $(w_1,w_2)$ where $r := w_2/w_1 \geq 1$, and that $\mathcal{D}$ is PDF-bounded. 
    With probability at least $1 - O(m^2/r)$, no WEF allocation exists.
    The same holds for WPROP.
\end{theorem}

\begin{proof}
    Since WEF and WPROP are equivalent for the case of $n = 2$, it suffices to focus on WEF.
    Suppose that $\mathcal{D}$ is $(\alpha,\beta)$-PDF-bounded, and let $Y=\min_{g\in M} u_2(g)$. 
    If $1 - \beta m/r \le 0$, then $1 - \beta m^2/r \le 0 \le \Pr[Y > m/r]$.
    Else, we have
    \begin{align*}
    \Pr\left[Y > \frac{m}{r}\right] &= \prod_{g\in M}\mathrm{Pr} \left[ u_2(g) > \frac{m}{r} \right] \\
    &\geq \left( 1 - \beta \cdot \frac{m}{r} \right)^m 
    \geq 1 - \beta \cdot \frac{m^2}{r},
    \end{align*}
    where we use the $(\alpha,\beta)$-PDF-boundedness of $\mathcal{D}$ for the first inequality and Bernoulli's inequality for the second.

    When $Y > m/r$, since any WEF allocation $(A_1, A_2)$ must give at least one item to agent $1$, we have $r \cdot u_2(A_1) \ge r \cdot Y > m > u_2(A_2)$. This means that no WEF allocation exists.
\end{proof}

\section{Conclusion and Future Work}

In conclusion, our work analyzes weighted fair division from an asymptotic perspective and establishes tight or asymptotically tight bounds on the existence of weighted envy-free and weighted proportional allocations.
Notably, a larger number of items is required for weighted proportionality than in the unweighted setting, and this number depends on the mean of the distribution from which the utilities are drawn.
We also investigated the relationship between the number of items and the weight ratio in the case of two agents.

Since we have considered settings where the weight ratio is bounded and where the number of agents is bounded separately, a natural next step is to obtain a more complete understanding when both parameters are allowed to grow.
For example, if all but one agents have the same (fixed) weight while the last agent's weight grows, what is the required number of items in terms of the number of agents and the last agent's weight?
Another interesting direction is to extend our results to a more general model where the items are allocated among \emph{groups} of agents rather than individual agents.
While asymptotic fair division for groups was previously explored \citep{ManurangsiSu17}, the study thus far has focused on groups with equal entitlements.
Permitting different entitlements can help us model applications where the groups have varying importance, e.g., depending on their sizes.


\section*{Acknowledgments}

This work was partially supported by JST
ERATO under grant number JPMJER2301, by the Singapore Ministry of Education under grant number MOE-T2EP20221-0001, and by an NUS Start-up Grant.
We thank the anonymous reviewers for their valuable feedback.

\bibliographystyle{named}
\bibliography{ijcai25}

\appendix

\onecolumn

\section{Proof of Lemma~\ref{lem:existence_of_T}}\label{appendix:proof_existence_of_T}

\subsection{Auxiliary Lemmas}
Before proving Lemma~\ref{lem:existence_of_T}, we establish a key property of the weighted picking sequence algorithm. 
Specifically, the following lemma shows that at any step~$s$, the number of picks $t_i(s)$ made by each agent $i$ is approximately proportional to her weight $w_i$.
Recall that $W = \sum_{i\in N}w_i$.
\begin{lemma}\label{lemma:t_i_s}
    In Algorithm~\ref{alg:WPS-mechanism}, for every step $s$ and all $i,j \in N$, we have
    \begin{equation}\label{eq:bound_diff}
        \left|\frac{t_i(s)}{w_i} - \frac{t_j(s)}{w_j} \right| \le \frac{1}{\min(w_i,w_j)}.
    \end{equation}
    Moreover, for every step $s$ and all $i\in N$, we have
    \begin{equation}\label{eq:t_i_s}
        \left|t_i(s) - \frac{w_i }{W} \cdot s  \right| \le \frac{w_i}{w_{\min}}.
    \end{equation}
\end{lemma}
\begin{proof}
    We first prove \eqref{eq:bound_diff} by induction on $s$.
    At $s=0$, we have $t_i(0) = 0$ for all $i\in N$, so $\left|\frac{t_i(0)}{w_i} - \frac{t_j(0)}{w_j} \right| = 0$ for all $i,j \in N$.
    Assume now that \eqref{eq:bound_diff} holds for some step $s\ge 0$, and let $i^* = \mathrm{argmin}_{i'\in N}\frac{t_{i'}(s)}{w_{i'}}$. 
    By description of the algorithm, at step $s + 1$, agent $i^*$ is selected. Thus, we have $t_{i^*}(s+1) = t_{i^*}(s) + 1$ and $t_{i'}(s+1) = t_{i'}(s)$ for all $i' \in N \setminus\{i^*\}$.
    This implies that for any pair $i,j\in N \setminus\{i^*\}$, \eqref{eq:bound_diff} holds at step $s+1$.
    Consider $i^*$ and any $i'\in N\setminus\{i^*\}$.
    If $\frac{t_{i^*}(s+1)}{w_{i^*}} \ge \frac{t_{i'}(s+1)}{w_{i'}} = \frac{t_{i'}(s)}{w_{i'}} $, then 
        \begin{align*}
        \left|\frac{t_{i^*}(s+1)}{w_{i^*}} - \frac{t_{i'}(s+1)}{w_{i'}} \right| 
        = \frac{t_{i^*}(s)}{w_{i^*}} + \frac{1}{w_{i^*}} - \frac{t_{i'}(s)}{w_{i'}}
        \le \frac{1}{w_{i^*}} \le \frac{1}{\min(w_{i^*},w_{i'})},
        \end{align*}
    where the first inequality follows from the definition of $i^*$.
    On the other hand, if $\frac{t_{i^*}(s+1)}{w_{i^*}} < \frac{t_{i'}(s+1)}{w_{i'}}$, then 
    \begin{align*}
    \left|\frac{t_{i^*}(s+1)}{w_{i^*}} - \frac{t_{i'}(s+1)}{w_{i'}}     \right| 
    = \frac{t_{i'}(s)}{w_{i'}} - \frac{t_{i^*}(s)}{w_{i^*}} - \frac{1}{w_{i^*}} 
        \le \frac{1}{\min(w_{i^*},w_{i'})} - \frac{1}{w_{i^*}} \le \frac{1}{\min(w_{i^*},w_{i'})},
    \end{align*}
    where the first inequality follows from the inductive hypothesis.
    Thus, we have established \eqref{eq:bound_diff}.
    
    Next, we prove \eqref{eq:t_i_s}. From \eqref{eq:bound_diff}, we have \begin{equation}\label{eq:t_i_s_over_w_i}
        \min_{i'\in N}\frac{t_{i'}(s)}{w_{i'}} \leq \frac{t_i(s)}{w_i}
        \leq \min_{i'\in N}\frac{t_{i'}(s)}{w_{i'}} + \frac{1}{w_{\min}}    
    \end{equation}
    for all $i\in N$. 
    This implies that
    $w_i\cdot \min_{i'\in N}\frac{t_{i'}(s)}{w_{i'}} \leq t_i(s)
        \leq w_i\cdot \left(\min_{i'\in N}\frac{t_{i'}(s)}{w_{i'}} + \frac{1}{w_{\min}}\right)$ for all $i\in N$. 
        Summing over all $i\in N$ and using the fact that $\sum_{i\in N} t_i(s) = s$, we get $W\cdot \min_{i'\in N}\frac{t_{i'}(s)}{w_{i'}} \leq s
        \leq W\cdot \left(\min_{i'\in N}\frac{t_{i'}(s)}{w_{i'}} + \frac{1}{w_{\min}}\right)$. It follows that
    \[
        \frac{s}{W} - \frac{1}{w_{\min}}
        \le
        \min_{i'\in N}\frac{t_{i'}(s)}{w_{i'}} \le \frac{s}{W}.
    \]
    Combining this with \eqref{eq:t_i_s_over_w_i}, we arrive at \eqref{eq:t_i_s}, as desired.
\end{proof}

We now state useful lemmas from the work of \citet{ManurangsiSu21}.
These lemmas are presented as Lemma~2.4, Proposition~2.2, and Lemma~2.6 in their work, respectively.

\begin{lemma}[\citep{ManurangsiSu21}]\label{lemma:r_c_d}
    Let $r,c,d$ be any positive integers such that $r\ge cd$, and let $X_1,X_2,\ldots, X_r$ be independent random variables in $[0,1]$ that are $(\alpha, \beta)$-PDF-bounded. Then,
    \[
        \mathrm{Pr}[X_1+X_2+ \cdots +X_r \ge r-c] \le 2^r \left(\frac{\beta}{d}\right)^{r-cd}.
    \]
\end{lemma}

\begin{lemma}[\citep{ManurangsiSu21}]\label{lemma:alpha_over_beta}
    For any $(\alpha, \beta)$-PDF-bounded distribution $\mathcal{D}$ and any $c\in (0,1]$, suppose that we draw $X$ from $\mathcal{D}_{\le c}$. Then, $X/c$ is generated from an $(\alpha/\beta, \beta/\alpha)$-PDF-bounded distribution.
\end{lemma}

\begin{lemma}[\citep{ManurangsiSu21}]\label{lemma:T_s_p}
    Let $T$ be a positive integer, let $s_1,s_2,\ldots, s_T$ be any positive integers, and let $s_{\min} = \min(s_1,s_2,\ldots, s_T)$. Consider the random variable $X_0=1$, and $X_1,X_2,\ldots,X_T$ such that $X_{k+1} \sim \mathcal{D}_{\le X_k}^{\max(s_{k+1})}$ for each $k \in \{0,1,\dots,T-1\}$. 
    If $\mathcal{D}$ is $(\alpha,\beta)$-PDF-bounded, then for any $p\in (0,1)$, we have
    \[
        \mathrm{Pr}\left[ X_T \ge 1 - \frac{\beta}{\alpha} \cdot \frac{T \log(T/p)}{s_{\min}} \right] \ge 1-p.
    \]
\end{lemma}

\subsection{Proof of (a) in Lemma~\ref{lem:existence_of_T}}
Let $T=  100 \left\lceil \tilde{\beta}\cdot\frac{w_i}{w_j}\cdot\frac{ \log m}{ \log \log m} \right\rceil \ge 100$. 
We will show that $\frac{t_i(m)}{2} \ge T $.
By Lemma~\ref{lemma:t_i_s}, we have
\begin{align*}
    \frac{t_i(m)}{2} - T 
    &\ge 
        \frac{w_i}{2W}\cdot m - \frac{w_i}{2w_{\min}} -  T  
    = 
        w_i \left(\frac{m}{2W}  - \frac{1}{2w_{\min}} -  \frac{T}{w_i}\right)  
    \ge 
        w_i \left(\frac{m}{2W} - \frac{1}{w_{\min}} -  \frac{T}{w_i}  \right).
\end{align*}
We will prove that $\frac{m}{2W}  - \frac{1}{w_{\min}} - \frac{T}{w_i} \ge 0$.
The inequality $C \ge \frac{w_{\max}}{w_{\min}}\ge \frac{W}{w_{\min}\cdot n}$ implies that $\frac{m}{4W} \ge \frac{m}{4C\cdot w_{\min}\cdot n}$.
Since the function $\frac{x}{\log x}$ is non-decreasing for $x \ge e$ (as its derivative is $\frac{\log x - 1}{\log^2 x}$), and we have
$\frac{m}{n} \ge 10^6 \tilde{\beta} \cdot C \cdot \frac{\log n}{ \log \log n}$, it follows that $\frac{m/n}{\log(m/n)} \ge 4\cdot 10^3 \tilde{\beta} \cdot C$ holds for sufficiently large $n$. Therefore,
\[
    \frac{m}{4W} 
    \ge 
        \frac{m}{4C\cdot w_{\min} \cdot n}
    \ge \frac{10^3 \tilde{\beta}}{w_{\min}}\log\left(\frac{m}{n}\right).
\]
This implies that
\begin{align}
    \frac{m}{2W}  - \frac{1}{w_{\min}} - \frac{T}{w_{i}}
    &= 
        \frac{m}{4W} + \frac{m}{4W}  - \frac{1}{w_{\min}}  - \frac{T}{w_{i}} \nonumber\\
    &\ge  
        \frac{10^6 \tilde{\beta}}{4 w_{\min}} \cdot \frac{\log n}{\log \log n} + \frac{10^3 \tilde{\beta}}
        {w_{\min}}\log\left(\frac{m}{n}\right) - \frac{1}{w_{\min}} - \frac{T}{w_{i}} \nonumber\\
    &\ge  
        \frac{10^3 \tilde{\beta}}{w_{\min}} \cdot \frac{ \log n}{\log \log n} + \frac{10^3 \tilde{\beta}}{w_{\min}}\log\left(\frac{m}{n}\right) - \frac{T}{w_{i}}  
        \tag{by $\frac{10^6 \tilde{\beta}}{4 w_{\min}} \cdot \frac{\log n}{ \log \log n} - \frac{1}{w_{\min}} \ge \frac{10^3 \tilde{\beta}}{w_{\min}} \cdot \frac{\log n}{\log \log n}$} \nonumber\\
    &\ge  
        \frac{10^3 \tilde{\beta}}{w_{\min}} \cdot \frac{ \log n}{\log \log m}+ \frac{10^3 \tilde{\beta}}{w_{\min}}\log\left(\frac{m}{n}\right)-   \frac{T}{w_{i}} 
        \tag{by $ m \ge n$} \nonumber\\
    &\ge   
        \frac{10^3 \tilde{\beta}}{w_{\min}} \cdot \frac{\log n}{\log \log m} 
        + \frac{10^3 \tilde{\beta}}{w_{\min}} \cdot \frac{\log (m/n)}{ \log \log m} - \frac{T}{w_{i}} 
        \tag{by $\log \log m \ge \log \log n \ge 1$} \nonumber\\
    &=  
        \frac{10^3 \tilde{\beta}}{w_{\min}} \cdot \frac{\log m}{\log \log m} 
        - \frac{T}{w_{i}} \nonumber\\
    &\ge  
        \frac{10^3 \tilde{\beta}}{w_{\min}} \cdot \frac{\log m}{\log \log m} 
        -\frac{100  \tilde{\beta}}{w_j} \cdot \frac{\log m}{\log \log m} - \frac{100}{w_i}
        \tag{by $T \le 100 \tilde{\beta} \cdot \frac{w_i}{w_j} \cdot \frac{ \log m}{ \log \log m}+100$} \nonumber\\
    &\ge  
        \frac{10^3 \tilde{\beta}}{w_{\min}} \cdot \frac{\log m}{\log \log m} 
        -\frac{100  \tilde{\beta}}{w_{\min}} \cdot \frac{\log m}{\log \log m} - \frac{100}{w_{\min}}
        \tag{by $w_i, w_j \ge w_{\min}$} \nonumber\\
    &\ge 
        0 \label{inequality:m_2W}.
\end{align}
Thus, we have $\frac{t_i(m)}{2} \ge T $.
This establishes that $T$ and $s^i(T)$ are well-defined, and $t_i(m)\ge T \ge 1$.

\subsection{Proof of (b) in Lemma~\ref{lem:existence_of_T}}
We show that with probability at least $1 - O(1/m^3)$, agent $i$ receives an item of value at least $1/2$ in her $T$-th pick. 
By Lemma~\ref{lemma:picking-generation}, $X_{1}, X_{2},\ldots, X_{T}$ are sampled according to
$
    X_{k} \sim \mathcal{D}_{\le X_{k-1}}^{\max(m - s^i(k) + 1)}
$
for each $k \in [T]$. 

First, from \eqref{inequality:m_2W} and Lemma~\ref{lemma:t_i_s}, we have
\begin{equation*}
    \frac{m}{2} - s^i(T) 
    \ge 
        \frac{m}{2} - W \left(\frac{T}{w_i} + \frac{1}{w_{\min}}\right) 
    = 
        W\left(\frac{m}{2W}  - \frac{1}{w_{\min}} - \frac{T}{w_i} \right) 
    \ge 
        0.
\end{equation*}
Thus, we can observe that at any step before the $T$-th pick of agent~$i$, at least half of the items remain available. 
Specifically, for any $k\in[T]$, it holds that $m -  s^i(k) + 1 \ge m -  s^i(T) + 1 \ge m- m/2 +1 \ge m/2$.

Next, since $T=   100 \left\lceil \tilde{\beta}\cdot \frac{w_i}{w_j}\cdot \frac{ \log m}{ \log \log m} \right\rceil$ and
$\frac{w_i}{w_j} \le C$ where $C \ge 1$, we obtain $T \le 200 \tilde{\beta}\cdot C\cdot \frac{ \log m}{ \log \log m}$. 
This implies that $\frac{\beta}{\alpha}\cdot \frac{T \log(T m^3)}{m/2} \le \frac{1}{2}$ for sufficiently large $m$.
Applying Lemma~\ref{lemma:T_s_p} with $p=1/m^3$, we obtain
\[
    \mathrm{Pr}\left[ X_{T} \ge \frac{1}{2} \right] 
    \ge 
        \mathrm{Pr}\left[ X_{T} \ge 1 - \frac{\beta}{\alpha} \cdot \frac{T \log(Tm^3)}{m/2} \right]
    \ge 
        \mathrm{Pr}\left[ X_{T} \ge 1 - \frac{\beta}{\alpha} \cdot \frac{T \log(T/p)}{m -  s^i(T) + 1} \right] 
    \ge 
        1-\frac{1}{m^3},
\]
as desired.

\subsection{Proof of (c) in Lemma~\ref{lem:existence_of_T}}
Our goal is to show that
\[
    \mathrm{Pr}\left[\sum_{k=1}^{T}  \sum_{\ell=1}^{\tau_{k}}\left(1-  \frac{X_{k,\ell}^{j}}{X_{k}} \right)\ge 2 \right] = 1-O\left(\frac{1}{m^3}\right).
\]
Since agent $i$ picks item $g_{T+1}^i$, by description of Algorithm~\ref{alg:WPS-mechanism}, we have
\begin{equation*}
    \frac{T}{w_i} \le \frac{\sum_{k=1}^T \tau_{k} +1}{w_j}.
\end{equation*}
It follows that
\begin{align*}
\sum_{k=1}^T \tau_{k}\ge \frac{w_j}{w_i}\cdot T -1  \ge \frac{100 \tilde{\beta}  \log m}{ \log \log m} -1 \ge \frac{80 \tilde{\beta}  \log m}{ \log \log m}.
\end{align*}
In particular,
\begin{align*}
\frac{8\tilde{\beta}}{\sum_{k=1}^T \tau_{k}}
\le \frac{\log\log m}{10\log m} \le 1.
\end{align*}

When conditioning on $X_{1} = x_1, \ldots, X_{T} = x_T$, 
the random variables 
$\frac{X^{j}_{1,1}}{X_{1}}, \frac{X^{j}_{1,2}}{X_{1}}, \ldots, \frac{X^{j}_{1,\tau_1}}{X_{1}}, \frac{X^{j}_{2,1}}{X_{2}}, \frac{X^{j}_{2,2}}{X_{2}}, \ldots, \frac{X^{j}_{2,\tau_2}}{X_{2}}, \ldots$, 
$\frac{X^{j}_{T,1}}{X_{T}}, \frac{X^{j}_{T,2}}{X_{T}}, \ldots, \frac{X^{j}_{T,\tau_T}}{X_{T}}$
are all independent.
By Lemma~\ref{lemma:alpha_over_beta}, these random variables are generated from an $(\alpha/\beta, \beta/\alpha)$-PDF-bounded distribution.
Therefore, applying Lemma~\ref{lemma:r_c_d} with $r = \sum_{k=1}^T \tau_{k}$, $c = 2$, and $d = \left\lfloor \frac{\sum_{k=1}^T \tau_{k}}{4} \right\rfloor \ge \frac{\sum_{k=1}^T \tau_{k}}{8}$,
we get
\begin{align*}
    & 
        \mathrm{Pr}\left[\sum_{k=1}^T \sum_{\ell=1}^{\tau_{k}} \frac{X^{j}_{k,\ell}}{X_{k}} 
        \ge \sum_{k=1}^T \tau_{k} - 2\,\middle|\, X_{1} = x_1, X_{2} = x_2, \ldots, X_{T} = x_T \right] \\
    &\le 
        2^{\sum_{k=1}^T \tau_{k}} \cdot\left(\frac{\tilde{\beta}}{\left\lfloor \frac{\sum_{k=1}^T \tau_{k}}{4} \right\rfloor}\right)^{\sum_{k=1}^T \tau_{k}-2\cdot\left\lfloor \frac{\sum_{k=1}^T \tau_{k}}{4} \right\rfloor} \\
    &\le 
        2^{\sum_{k=1}^T \tau_{k}} \cdot\left(\frac{8\tilde{\beta}}{\sum_{k=1}^T \tau_{k}}\right)^{\sum_{k=1}^T \tau_{k}-2\cdot\left\lfloor \frac{\sum_{k=1}^T \tau_{k}}{4} \right\rfloor} \\
    &\le 
        2^{\sum_{k=1}^T \tau_{k}} \cdot\left(\frac{8\tilde{\beta}}{\sum_{k=1}^T \tau_{k}}\right)^{\sum_{k=1}^T \tau_{k}-2\cdot\frac{\sum_{k=1}^T \tau_{k}}{4}} \tag{by $\frac{8\tilde{\beta}}{\sum_{k=1}^T \tau_{k}} \le 1$} \\
    &= 
        \left(\frac{32\tilde{\beta} }{\sum_{k=1}^T \tau_{k}}\right)^{\frac{1}{2}\sum_{k=1}^T \tau_{k} } \\
    &\le 
        \left(\frac{32\tilde{\beta} }{\frac{80 \tilde{\beta}  \log m}{ \log \log m}}\right)^{\frac{1}{2}\cdot\frac{80 \tilde{\beta}  \log m}{ \log \log m} } \\
    &= 
        \left(\frac{32\log \log m }{80   \log m}\right)^{\frac{1}{2}\cdot\frac{80 \tilde{\beta}  \log m}{ \log \log m} } \\
    &\le 
        \left(\frac{ \log \log m}{ \log m}\right)^{\frac{40 \tilde{\beta}  \log m}{ \log \log m} } \\
    &\le 
        \left(\frac{1}{ \sqrt{\log m}}\right)^{\frac{6 \log m}{ \log \log m} } \tag{by $\log \log m \le \sqrt{\log m}$}\\
    &= 
        \frac{1}{m^3}.
\end{align*}
Taking the expectation with respect to all possible values of $x_1,x_2,\ldots,x_T$, we get
\[
    \mathrm{Pr}\left[\sum_{k=1}^T \sum_{\ell=1}^{\tau_{k}} \frac{X^{j}_{k,\ell}}{X_{k}}  
    \ge \sum_{k=1}^T \tau_{k} - 2 \right] 
    = O\left(\frac{1}{m^3}\right).
\]
Hence, with probability at least $1-O\left(\frac{1}{m^3}\right)$, it holds that
\begin{align*}
    \sum_{k=1}^{T}  \sum_{\ell=1}^{\tau_{k}}\left(1-  \frac{X_{k,\ell}^{j}}{X_{k}} \right)
    &\ge \sum_{k=1}^T \tau_{k} -  \left(\sum_{k=1}^T \tau_{k} - 2 \right) = 2,
\end{align*} 
completing the proof.

\section{Proof of Lemma~\ref{lemma:existence_of_perfect_s_matching_in_random_graph}}\label{appendix:proof_existence_of_perfect_s_matching_in_random_graph}

    By Proposition~\ref{proposition:a_perfect_s-matching}, if there is no left-saturating $\bm{s}$-matching, then there must exist a subset $Y\subseteq L$ such that $|N_G(Y)| < \sum_{i\in Y} s_i$. 
    For any $k\in[n]$, let $\Gamma_k$ be the sum of the $k$ largest elements of $\bm{s}$,
    and $\gamma_k$ be the sum of the $k$ smallest elements of $\bm{s}$.
    By the union bound, the probability that there is no left-saturating $\bm{s}$-matching can be bounded as follows:
    \begin{align*}
        &\mathrm{Pr}\left[ \text{No left-saturating $\bm{s}$-matching exists} \right] \\
        &\le 
            \mathrm{Pr}\left[ \text{$\exists Y\subseteq L$ such that $|N_G(Y)| < \sum_{i\in Y} s_i$} \right] \\
        &\le 
            \sum_{k=1}^n \mathrm{Pr}\left[ \text{$\exists Y\subseteq L$ such that $|Y|=k, |N_G(Y)| < \sum_{i\in Y} s_i$} \right] \\
        &= 
            \sum_{k=1}^n \mathrm{Pr}\left[ \text{$\exists Y\subseteq L$, $Z\subseteq R$ such that $|Y|=k, |Z|=\sum_{i\in Y} s_i-1, N_G(Y) \subseteq Z$} \right] \\
        &\le 
            \sum_{k=1}^n \sum_{\substack{Y\subseteq L,\, Z\subseteq R \\ |Y|=k,\, |Z|=\sum_{i\in Y} s_i-1}}\mathrm{Pr}\left[ N_G(Y) \subseteq Z \right] \\
        &= 
            \sum_{k=1}^n \sum_{\substack{Y\subseteq L,\, Z\subseteq R \\ |Y|=k,\, |Z|=\sum_{i\in Y} s_i-1}}
            (1-p)^{k\cdot \left(m-\sum_{i\in Y} s_i+1\right)} \\
        &\le 
            \sum_{k=1}^n \sum_{\substack{Y\subseteq L, \\ |Y|=k}} {m \choose m-\sum_{i\in Y} s_i+1}\,
            \mathrm{exp}\left(-p k\cdot \left(m-\sum_{i\in Y} s_i +1\right)\right) \tag{by $(1-p)^{k \ell} \le \mathrm{exp}(-pk\ell)$ for any $\ell\ge 0$}\\
        &\le 
            \sum_{k=1}^n \sum_{\substack{Y\subseteq L, \\ |Y|=k}} m^{\min\left\{\sum_{i\in Y} s_i,m-\sum_{i\in Y} s_i+1 \right\}}\cdot
            \mathrm{exp}\left(-p k\cdot \left(m-\sum_{i\in Y} s_i +1\right)\right) \\
        &= 
            \sum_{k=1}^n \sum_{\substack{Y\subseteq L, \\ |Y|=k}} 
            \mathrm{exp}\left(\min\left\{\sum_{i\in Y} s_i,m-\sum_{i\in Y} s_i+1 \right\}\cdot \log m 
            -p k\cdot \left(m-\sum_{i\in Y} s_i +1\right)\right) \\
        &= 
            \sum_{k=1}^n \sum_{\substack{Y\subseteq L, \\ |Y|=k}} 
            \mathrm{exp}\left(\min\left\{\sum_{i\in Y} s_i,m-\sum_{i\in Y} s_i+1 \right\}\cdot \log m \right.\\ 
        &\quad\left.
            -p \cdot \frac{k}{\sum_{i\in Y} s_i} \cdot \max\left\{\sum_{i\in Y} s_i, m-\sum_{i\in Y} s_i+1 \right\}\cdot \min\left\{\sum_{i\in Y} s_i, m-\sum_{i\in Y} s_i+1 \right\}\right) \\
        &\le 
            \sum_{k=1}^n 
            \sum_{\substack{Y\subseteq L, \\ |Y|=k}} 
            \mathrm{exp}\left(\left(\log m-p \cdot\frac{k}{\sum_{i\in Y} s_i} \cdot \max\left\{\sum_{i\in Y} s_i, m-\sum_{i\in Y} s_i+1 \right\}\right) \cdot\min\left\{\sum_{i\in Y} s_i,m-\sum_{i\in Y} s_i+1 \right\}\right).
    \end{align*}
For $Y\subseteq L$ with $|Y|=k$, observe that $\max\left\{\sum_{i\in Y} s_i, m-\sum_{i\in Y} s_i+1 \right\} \ge \frac{m}{2} \ge \frac{n s_{\min}}{2}$ and $\frac{k}{\sum_{i\in Y} s_i} \ge \frac{1}{s_{\max}}$. These inequalities yield
\[
    \log m-p \cdot \frac{k}{\sum_{i\in Y} s_i} \cdot \max\left\{\sum_{i\in Y} s_i, m-\sum_{i\in Y} s_i+1 \right\} 
    \le \log m-p \cdot \frac{1}{s_{\max}} \cdot \frac{n s_{\min}}{2} \le \log m-p \cdot \frac{n}{2B}
    \le - 3\log m,
\]
where we use the assumption $p \ge 8 B \log m / n$ for the last inequality.
Moreover, since $s_i \ge 1$ for all $i\in N$, we have $k \le \gamma_k$ and $n-k \le \Gamma_n-\Gamma_k \le m-\Gamma_k$ for all $k\in [n]$.
Combining these bounds, we deduce that
\begin{align*}
    &\mathrm{Pr}\left[ \text{No left-saturating $\bm{s}$-matching exists} \right] \\
    &\le \sum_{k=1}^n 
        \sum_{\substack{Y\subseteq L, \\ |Y|=k}} 
        \mathrm{exp}\left(- 3\log m \cdot
        \min\left\{\gamma_k,m-\Gamma_k+1 \right\}\right) \\
    &\le \sum_{k=1}^n 
        {n \choose k}\,
        \mathrm{exp}\left( - 3\log m \cdot
        \min\left\{\gamma_k,m-\Gamma_k+1 \right\}\right) \\
    &\le \sum_{k=1}^n 
        n^{\min\{k, n-k\}}\cdot
        \mathrm{exp}\left( - 3\log m \cdot
        \min\left\{\gamma_k,m-\Gamma_k+1 \right\}\right) \\
    &= \sum_{k=1}^n 
        \mathrm{exp}\left(\log n\cdot \min\left\{k,n-k \right\}\right)\cdot
        \mathrm{exp}\left( - 3\log m \cdot
        \min\left\{\gamma_k,m-\Gamma_k+1 \right\}\right) \\
    &\le \sum_{k=1}^n 
        \mathrm{exp}\left( \log m \cdot \min\left\{k,n-k \right\} - 3\log m \cdot
        \min\left\{\gamma_k,m-\Gamma_k+1 \right\}\right) \tag{by $\log n \le \log m$}\\
    &\le \sum_{k=1}^n 
        \mathrm{exp}\left( - 2\log m \cdot
        \min\left\{\gamma_k,m-\Gamma_k+1 \right\}\right) \tag{by $k \le \gamma_k$ and $n-k \le m-\Gamma_k+1$}\\
    &\le \sum_{k=1}^n 
        \mathrm{exp}\left( - 2\log m \right) \tag{by $\min\left\{\gamma_k,m-\Gamma_k+1 \right\} \ge 1$}\\
    &\le \frac{1}{n} \tag{by $n\le m$}.
\end{align*}
Therefore, the probability that $G$ contains a left-saturating $\bm{s}$-matching is at least $1-1/n$, which approaches $1$ as $n \to \infty$.

\end{document}